\documentclass[12pt,letterpaper]{article}
\usepackage{amsmath,amsfonts,amsthm,amssymb,stmaryrd,bm}
\usepackage{graphicx,relsize,epic}

\theoremstyle{plain}

\numberwithin{equation}{section}
\newtheorem{thm}{Theorem}[section]
\newtheorem{cor}[thm]{Corollary}

\newenvironment{exam}[1]%
{\begin{flushleft}\textbf{Example #1}.\enspace}%
{\end{flushleft}}

\newcommand{\positive}{{\mathbb N}}
\newcommand{\complex}{{\mathbb C}}
\newcommand{\real}{{\mathbb R}}
\newcommand{\ascript}{{\mathcal A}}
\newcommand{\cscript}{{\mathcal C}}
\newcommand{\pscript}{{\mathcal P}}
\newcommand{\sscript}{{\mathcal S}}
\newcommand{\rmcyl}{\mathrm{cyl}}
\newcommand{\rmre}{\mathrm{Re}}
\newcommand{\chat}{\widehat{c}}
\newcommand{\atilde}{\widetilde{a}}
\newcommand{\cbar}{\overline{c}}
\newcommand{\alphabar}{\overline{\alpha}}
\newcommand{\iunderbar}{\underline{i}}
\newcommand{\junderbar}{\underline{j}}
\newcommand{\kunderbar}{\underline{k}}
\newcommand{\ctimes}{\mathrel{\mathlarger\cdot}}
\newcommand{\ouparrow}{\mathord{\shortuparrow}}

\newcommand{\ab}[1]{\left|#1\right|}
\newcommand{\brac}[1]{\left\{#1\right\}}
\newcommand{\paren}[1]{\left(#1\right)}
\newcommand{\sqbrac}[1]{\left[#1\right]}
\newcommand{\elbows}[1]{{\left\langle#1\right\rangle}}
\newcommand{\ket}[1]{{\left|#1\right>}}
\newcommand{\bra}[1]{{\left<#1\right|}}

\begin{document}

\title{THE UNIVERSE AS\\A QUANTUM COMPUTER
}
\author{S. Gudder\\ Department of Mathematics\\
University of Denver\\ Denver, Colorado 80208, U.S.A.\\
sgudder@du.edu
}
\date{}
\maketitle

\begin{abstract}
This article presents a sequential growth model for the universe that acts like a quantum computer. The basic constituents of the model are a special type of causal set (causet) called a $c$-causet. A $c$-causet is defined to be a causet that is independent of its labeling. We characterize $c$-causets as those causets that form a multipartite graph or equivalently those causets whose elements are comparable whenever their heights are different. We show that a $c$-causet has precisely two $c$-causet offspring. It follows that there are $2^n$ $c$-causets of cardinality $n+1$. This enables us to classify $c$-causets of cardinality $n+1$ in terms of $n$-bits. We then quantize the model by introducing a quantum sequential growth process. This is accomplished by replacing the $n$-bits by $n$-qubits and defining transition amplitudes for the growth transitions. We mainly consider two types of processes called stationary and completely stationary. We show that for stationary processes, the probability operators are tensor products of positive rank-1 qubit operators. Moreover, the converse of this result holds. Simplifications occur for completely stationary processes. We close with examples of precluded events.
\end{abstract}

\section{Introduction}  
One frequently hears people say that the universe acts like a giant quantum computer, but when pressed they are usually short on details. This article attempts to begin giving these details. It should be emphasized that only a basic framework is presented and much work remains to be done. If this idea is correct, then great benefits will result. One benefit being better understanding of the universe itself and another is the ability to tap into a source of enormous computational power.

We first present a theory of discrete quantum gravity in terms of causal sets (causets) \cite{gud13, sor03,sur11}. Unlike previous sequential growth models the basic elements of this theory are a special type of causet called a covariant causet ($c$-causet). A
$c$-causet is defined to be a causet that is independent of its labeling. That is, two different labelings of a $c$-causet are isomorphic. The restriction of a growth model to $c$-causets provides great simplifications. For example, every $c$-causet possesses a unique
$c$-causet history and has precisely two covariant offspring. It follows that there are $2^n$ $c$-causets of cardinality $n+1$. This enables us to classify $c$-causets of cardinality $n+1$ in terms of $n$-bits. The framework of a classical computer is already emerging. We characterize $c$-causets as those causets that form a multipartite graph or equivalently those causets whose elements are comparable whenever their heights are different.

We next quantize the model by introducing a quantum sequential growth process. This is accomplished by replacing the $n$-bits with
$n$-qubits and defining transition amplitudes for the growth transitions. The transition amplitudes are given by complex-valued coupling constants $c_{n,j}$, $j=0,1,\ldots,2^{n-1}$. If the coupling constants are independent of $j$, we call the process stationary and if they are independent of $n$ and $j$ we call the process completely stationary. We show that for stationary processes the probability operators that determine the quantum dynamics are tensor products of rank-1 qubit operators. Moreover, the converse of this result holds. Simplifications occur for completely stationary processes. In this case, all the qubit operators are the same and can be related to spin operators. We close with some examples of precluded events in the completely stationary case.

\section{Covariant Causets} 
In this article we call a finite partially ordered set a \textit{causet}. If two causets are order isomorphic, we consider them to be identical. If $a$ and $b$ are elements of a causet $x$, we interpret the order $a<b$ as meaning that $b$ is in the causal future of $a$ and $a$ is in the causal past of $b$. An element $a\in x$ is \textit{maximal} if there is no $b\in x$ with $a<b$. If $a<b$ and there is no $c\in x$ with $a<c<b$, then $a$ is a \textit{parent} of $b$ and $b$ is a \textit{child} of $a$. If $a,b\in x$ we say that $a$ and $b$ are \textit{comparable} if $a\le b$ or $b\le a$. A \textit{chain} in $x$ is a set of mutually comparable elements of $x$ and an \textit{antichain} is a set of mutually incomparable elements of $x$. The \textit{height} of $a\in x$ is the cardinality of the longest chain whose largest element is $a$. The height of $x$ is the maximum of the heights of its elements. We denote the cardinality of $x$ by $\ab{x}$.

If $x$ and $y$ are causets with $\ab{y}=\ab{x}+1$, then $x$ \textit{produces} $y$ if $y$ is obtained from $x$ by adjoining a single maximal element $a$ to $x$. In this case we write $y=x\shortuparrow a$ and use the notation $x\to y$. If $x\to y$, we also say that $x$ is a \textit{producer} of $y$ and $y$ is an \textit{offspring} of $x$. In general, $x$ may produce many offspring and $y$ may be the offspring of many producers.

A \textit{labeling} for a causet $x$ is a bijection $\ell\colon x\to\brac{1,2,\ldots ,\ab{x}}$ such that $a,b\in x$ with $a<b$ implies that
$\ell (a)<\ell (b)$. A \textit{labeled causet} is a pair $(x,\ell )$ where $\ell$ is a labeling of $x$. For simplicity, we frequently write
$x=(x,\ell )$ and call $x$ an $\ell$-\textit{causet}. Two $\ell$-causets $x$ and $y$ are \textit{isomorphic} if there exists a bijection
$\phi\colon x\to y$ such that $a<b$ if and only if $\phi (a)<\phi (b)$ and $\ell\sqbrac{\phi (a)}=\ell (a)$ for every $a\in x$. Isomorphic
$\ell$-causets are considered identical as $\ell$-causets. It is not hard to show that any causet can be labeled in many different ways but there are exceptions and these are the ones of importance in this work. A causet is \textit{covariant} if it has a unique labeling (up to $\ell$-causet isomorphism). Covariance is a strong restriction which says that the elements of the causet have a unique ``birth order'' up to isomorphism. We call a covariant causet a $c$-coset.

We denote the set of $c$-causets with cardinality $n$ by $\pscript _n$ and the set of all $c$-causets by $\pscript =\cup\pscript _n$. Notice that any nonempty $c$-causet $y$ has a unique producer. Indeed, if $y$ had two different producers $x_1,x_2$ then $x_1$ and $x_2$ could be labeled differently and these could be used to give different labelings for $y$. If $x\in\pscript$, then the parent-child relation $a\prec b$ makes $x$ into a graph $(x,\prec )$. A graph $G$ is \textit{multipartite} if there is a partition of its vertices
$V=\cup V_j$ such that the vertices of $V_j$ and $V_{j+1}$ are adjacent and there are no other adjacencies. 

\begin{thm}       
\label{thm21}
The following statements for a causet $x$ are equivalent.
{\rm (a)}\enspace $x$ is covariant,
{\rm (b)}\enspace the graph $(x,\prec )$ is multipartite,
{\rm (c)}\enspace $a,b\in x$ are comparable whenever $a$ and $b$ have different heights.
\end{thm}
\begin{proof}
Conditions (b) and (c) are clearly equivalent. To prove that (a) implies (b), suppose $x$ is covariant and let $x=\cup _{i=0}^my_i$ where $y_i$ is the set of elements in $x$ of height $i$. Suppose $a\in y_n$, $b\in y_{n+1}$ and $a\not< b$. We can delete maximal elements of $y$ until $b$ is maximal and the only element of height $n+1$. Denote the resulting causet by $z$. We can label $b$ by $\ab{z}$,
$a$ by $\ab{z}-1$ and consistently label the other elements of $z$ so that $z$ is an $\ell$-causet. We can also label $b$ by $\ab{z}-1$, $a$ by $\ab{z}$ and keep the same labels for the other elements of $z$. This gives two nonisomorphic labelings of $z$. Adjoining maximal elements to $z$ to obtain $x$, we have $x$ with two nonisomorphic labelings which is a contradiction. Hence, $a<b$ so $a$ is a parent of $b$. It follows that $x$ is multipartite. To prove that (b) implies (a), suppose the graph $(x,\prec )$ is multipartite. Letting
$x=\cup _{i=0}^my_i$ where $y_i$ is the set of elements of height $i$, it follows that $a<b$ for all $a\in y_i$, $b\in y_{i+1}$, $i=0,\ldots ,m-1$. We can write
\begin{align*}
y_0&=\brac{a_1,\ldots ,a_{\ab{y_0}}}\\
y_1&=\brac{a_{\ab{y_0}+1},\ldots ,a_{\ab{y_0}+\ab{y_1}}}\\
\vdots&\\
y_m&=\brac{a_{\ab{y_0}+\cdots +\ab{y_{m-1}}+1},\ldots ,a_{\ab{y_0}+\cdots +\ab{y_m}}}
\end{align*}
where $j$ is the label on $a_j$. This gives a labeling of $x$ and is the only labeling up to isomorphism.
\end{proof}

\begin{thm}       
\label{thm22}
If $x\in\pscript$, then $x$ has precisely two covariant offspring.
\end{thm}
\begin{proof}
By Theorem~\ref{thm21}, the graph $(x,\prec )$ is multipartite. Suppose $x$ has height $n$. Let $x_1=x\shortuparrow a$ where $a$ has all the elements of height $n$ as parents. Then $a$ is the only element of $x_1$ with height $n+1$. Hence, $x_1$ is multipartite so by Theorem~\ref{thm21}, $x_1$ is a covariant offspring of $x$. Let $x_2=x\shortuparrow b$ where $b$ has all the elements of height $n-1$ in $x$ as parents. (If $n=1$, then $b$ has no parents.) It is clear that $x_2$ is a multipartite graph. By Theorem~\ref{thm21}, $x_2$ is a covariant offspring of $x$. Also, there is only one covariant offspring of each of these two types. Let $y=x\shortuparrow c$ be a covariant offspring of $x$ that is not one of these two types and let $a\in x$ have label $\ab{x}$. Then $a$ and $c$ are incomparable and we can label $x$ by $\ab{x}+1$. If we interchange the labels of $a$ and $c$, we get a nonisomorphic labeling of $y$ which gives a contradiction. We conclude that $x$ has precisely two covariant offspring.
\end{proof}

\begin{cor}       
\label{cor23}
There are $2^n$ $c$-causets of cardinality $n+1$.
\end{cor}
\begin{proof}
Notice that we obtain all $c$-causets from the producer-offspring process of Theorem~\ref{thm22}. Indeed, take any $x\in\pscript$ and delete maximal elements until we arrive at the one element $c$-causet. In this way, $x$ is obtained from the process of
Theorem~\ref{thm22}. We now employ induction on $n$. There are $1=2^{1-1}$ $c$-causets of cardinality 1. If the result holds for
$c$-causets of cardinality $n$, then by Theorem~\ref{thm22} there are $2\ctimes 2^{n-1}=2^n$ $c$-causets of cardinality $n+1$. Hence, the result holds for $c$-causets of cardinality $n+1$.
\end{proof}

As a bonus we obtain an already known combinatorial identity. A \textit{composition} of a positive integer $n$ is a sequence of positive integers whose sum is $n$. The order of terms in the sequence is taken into account. For example the following are the compositions of $1,2,3,4,5$.
\begin{align*}
n&=1\colon 1\\
n&=2\colon 1+1,2\\
n&=3\colon 1+1+1,1+2,2+1,3\\
n&=4\colon 1+1+1+1,1+1+2,1+2+1,2+1+1,2+2,1+3,3+1,4\\
n&=5\colon 1+1+1+1+1,1+1+1+2, 1+1+2+1,1+2+1+1,\\
&\hskip 2.6pc 2+1+1+1,1+1+3,1+3+1,3+1+1,1+4,4+1,\\
&\hskip 2.6pc 2+3,3+2,1+2+2,2+1+2,2+2+1,5
\end{align*}
The reader has surely noticed that for $n=1,2,3,4,5$, the number of compositions of $n$ is $2^{n-1}$.

\begin{cor}       
\label{cor24}
There are $2^{n-1}$ compositions of the positive integer $n$.
\end{cor}
\begin{proof}
There is a bijection between compositions of $n$ and multipartite graphs with $n$ vertices. The result follows from
Corollary~\ref{cor23}.
\end{proof}

\begin{figure}
\vglue -8pc
\includegraphics[scale=.85]{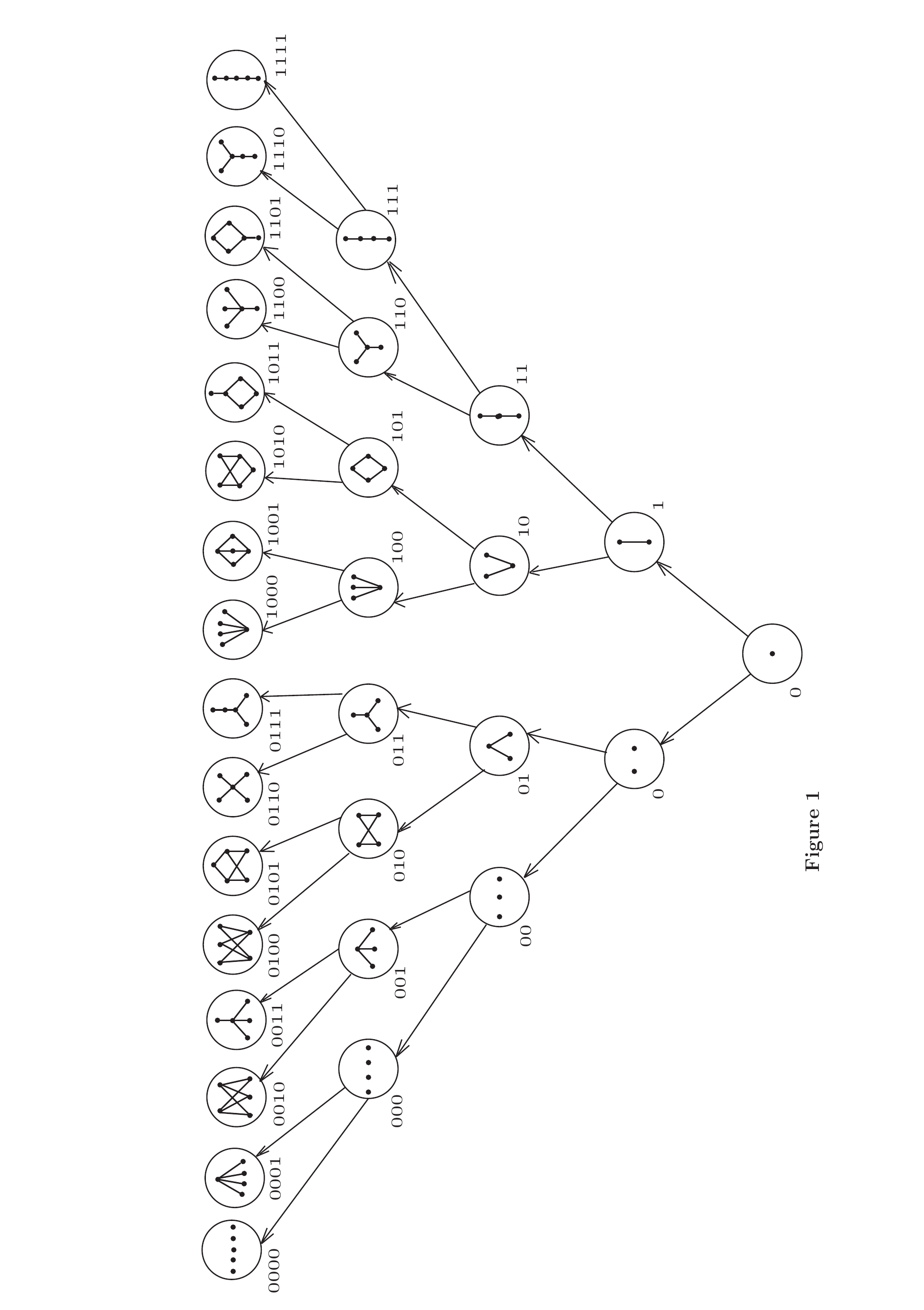}
\end{figure}
\medskip

The pair $(\pscript, \to )$ forms a partially ordered set in its own right. Moreover, $(\pscript ,\to )$ also forms a graph that is a tree. Figure~1 depicts the first five levels of this tree. The binary designations in Figure~1 will now be explained. By Corollary~\ref{cor23}, at height $n+1$ there are $2^n$ $c$-causets so binary numbers fit well, but how do we define a natural order for the $c$-causets? We have seen in Theorem~\ref{thm22} that if $x\in\pscript _n$, $n=1,2,\ldots$, then $x$ has precisely two offspring in $\pscript$,
$x\to x_0,x_1$ here $x_0$ has the same height as $x$ and $x_1$ has the height of $x$ plus one. We call $x_0$ the
0-\textit{offspring} and $x_1$ the 1-\textit{offspring} of $x$. We assign a \textit{binary order} to $x\in\pscript$ recursively as follows. If
$x\in\pscript _1$, then $x$ is the unique one element $c$-causet and we designate $x$ by 0. If $x\in\pscript _{n+1}$, then $x$ has a unique producer $y\in\pscript _n$. Suppose $y$ has binary order $j_{n-1}j_{n-1}\cdots j_2j_1$, $j_i=0$ or 1. If $x$ is the 0-offspring of $y$, then we designate $x$ with $j_{n-1}\cdots j_2j_10$ and if $x$ is a 1-offspring of $y$, then we designate $x$ with
$j_{n-1}\cdots j_2j_11$. The reader can now check this definition with the binary order in Figure~1.

We now see the beginning development of a giant classical computer. At the $(n+1)$th step of the process, $n$-bit strings are generated. It is estimated that we are now at about the $10^{60}$th step so $(10^{60}-1)$-bit strings are being generated. There are about $2^{10^{60}}$ such strings so an enormous amount of information is being processed. When we get to quantum computers, then superpositions of strings will be possible and the amount of information increases exponentially. It is convenient to employ the notation
\begin{equation*}
\junderbar = j_nj_{n-1}\cdots j_2j_1
\end{equation*}
for an $n$-bit string. In this way we can designate each $x\in\pscript$ uniquely by $x_{n+1,\junderbar}$ where $n+1=\ab{x}$. For example, the $c$-causets at step~3 in Figure~1 are $x_{3,00},x_{3,01},x_{3,10},x_{3,11}$. In decimal notation we can also write these as $x_{3,0},x_{3,1}x_{3,2},x_{3,3}$.

The binary order that we have just discussed in equivalent to a natural order in terms of the $c$-causet structure. Let
$x=\brac{a_1,\ldots ,a_n}\in\pscript _n$ where we can assume without loss of generality that $j$ is the label of $a_j$, $j=1,\ldots ,n$. Define
\begin{equation*}
j_x\ouparrow =\brac{i\in\positive\colon a_j<a_i}
\end{equation*}
Thus, $j_x\ouparrow$ is the set of labels of the descendants of $a_j$. Order the set of $c$-cosets in $\pscript _n$ lexicographically as follows. If $x,y\in\pscript _n$, then $x<y$ if
\begin{equation*}
1_x\ouparrow =1_y\ouparrow,\cdots ,j_x\ouparrow =j_y\ouparrow , (j+1)_x\subsetneq\ouparrow (j+1)\ouparrow
\end{equation*}
It is easy to check that $<$ is a total order relation on $\pscript _n$. The next theorem, whose proof we leave to the reader, shows that the order $<$ on $\pscript _n$ is equivalent to the binary order previously discussed.

\begin{thm}       
\label{thm25}
If $x_{n,\junderbar},x_{n,\kunderbar}\in\pscript _n$, then $x_{n\junderbar}<x_{n,\kunderbar}$ if and only if $\junderbar <\kunderbar$.
\end{thm}

\begin{exam}{1}  
We can illustrate Theorem~\ref{thm25} by considering $\pscript _4$. For the $c$-causets $x_{4,0},x_{4,1},\ldots ,x_{4,7}\in\pscript _4$ we list the sets $(1_x\ouparrow ,2_x\ouparrow ,3_x\ouparrow )$. Notice that we need not list $4_x\ouparrow=\emptyset$ in all cases of
$\pscript _4$.
\begin{align*}
x_{4,0}&\colon (\emptyset ,\emptyset ,\emptyset )\\
x_{4,1}&\colon (\brac{4},\brac{4},\brac{4})\\
x_{4,2}&\colon (\brac{3,4},\brac{3,4},\emptyset )\\
x_{4,3}&\colon (\brac{2,3,4},\brac{3,4},\brac{4})\\
x_{4,4}&\colon (\brac{2,3,4},\emptyset ,\emptyset )\\
x_{4,5}&\colon (\brac{2,3,4},\brac{4},\brac{4})\\
x_{4,6}&\colon (\brac{2,3,4},\brac{3,4},\emptyset )\\
x_{4,7}&\colon (\brac{2,3,4},\brac{3,4},\brac{4})
\end{align*}
The lexicographical order becomes:
\begin{equation*}
x_{4,0}<x_{4,1}<x_{4,2}<x_{4,3}<x_{4,4}<x_{4,5}<x_{4,6}<x_{x,7}
\end{equation*}
\end{exam}

\begin{exam}{2}   
This is so much fun that we list the sets
\begin{equation*}
(1_x\ouparrow ,2_x\ouparrow ,3_x\ouparrow ,4_x\ouparrow )
\end{equation*}
for the $c$-causets $x_{5,0},\ldots ,x_{5,15}\in\pscript _5$.
\begin{align*}
x_{5,0}&\colon (\emptyset ,\emptyset ,\emptyset ,\emptyset )\hskip 10pc x_{5,1}\colon (\brac{5},\brac{5},\brac{5},\brac{5})\\
x_{5,2}&\colon (\brac{4,5},\brac{4,5},\brac{4,5}, \emptyset )\hskip 3.75pc x_{5,3}\colon (\brac{4,5},\brac{4,5},\brac{4,5},\brac{5})\\
x_{5,4}&\colon (\brac{3,4,5},\brac{3,4,5},\emptyset , \emptyset )\hskip 4pc x_{5,5}\colon (\brac{3,4,5},\brac{3,4,5},\brac{5},\brac{5})\\
x_{5,6}&\colon (\brac{3,4,5},\brac{3,4,5},\brac{4,5}, \emptyset )\hskip 2pc x_{5,7}\colon (\brac{3,4,5},\brac{3,4,5},\brac{4,5},\brac{5})\\
x_{5,8}&\colon (\brac{2,3,4,5},\emptyset ,\emptyset , \emptyset )\hskip 6pc x_{5,9}\colon (\brac{2,3,4,5},\brac{5},\brac{5},\brac{5})\\
x_{5,10}&\colon (\brac{2,3,4,5},\brac{4,5},\brac{4,5}, \emptyset )
  \hskip 1.75pc x_{5,11}\colon (\brac{2,3,4,5},\brac{4,5},\brac{4,5},\brac{5})\\
x_{5,12}&\colon (\brac{2,3,4,5},\brac{3,4,5},\emptyset, \emptyset )
  \hskip 3pc x_{5,13}\colon (\brac{2,3,4,5},\brac{3,4,5},\brac{5},\brac{5})\\
x_{5,14}&\colon (\brac{2,3,4,5},\brac{3,4,5},\brac{4,5}, \emptyset )
  \hskip 1pc x_{5,15}\colon (\brac{2,3,4,5},\brac{3,4,5},\brac{4,5},\brac{5})
\end{align*}

This order structure $(\pscript _n,<)$ induces a topology on $\pscript _n$ whereby we can describe the ``closeness'' of $c$-causets. For example, we can place a metric on $\pscript _n$ by defining $\rho (x_{n\junderbar},x_{n,\kunderbar})=\ab{\junderbar -\kunderbar}$. If we want to keep the size of the metric reasonable, we could define
\begin{equation*}
\rho (x_{n,\junderbar},x_{n,\kunderbar})=\frac{1}{2^{n-1}}\ab{\junderbar -\kunderbar}
\end{equation*}
\end{exam}

\section{Quantum Sequential Growth Processes} 
The tree $(\pscript ,\to )$ can be thought of as a growth model and an $x\in\pscript _n$ is a possible universe at step (time) $n$. An instantaneous universe $x$ grows one element at a time in one of two ways at each step. A \textit{path} in $\pscript$ is a sequence (string) $\omega _1\omega _2\cdots$ where $\omega _i\in\pscript _i$ and $\omega _i\to\omega _{i+1}$. An $n$-\textit{path} is a finite sequence $\omega _1\omega _2\cdots\omega _n$ where again $\omega _i\in\pscript _i$ and $\omega _i\to\omega _{i+1}$. We denote the set of paths by $\Omega$ and the set of $n$-paths by $\Omega _n$. We think of $\omega\in\Omega$ as a ``completed'' universe or as a universal history. We may also view $\omega\in\Omega$ as an evolving universe. Since a $c$-causet has a unique producer, an $n$-path $\omega =\omega _1\omega _2\cdots\omega _n$ is completely determined by $\omega _n$. In other words, a $c$-causet possesses a unique history. We can thus identify $\Omega _n$ with $\pscript _n$ and we write $\Omega _n\approx\pscript _n$. If
$\omega =\omega _1\omega _2\cdots\omega _n\in\Omega _n$ we denote by $\omega\to )$ the two element subset of
$\Omega _{n+1}$ consisting of $\brac{\omega x_0,\omega x_1}$ where $x_0$ and $x_1$ are the offspring of $\omega _n$. Thus,
\begin{equation*}
(\omega\to )=\brac{\omega _1\cdots\omega _nx_0,\omega _1\cdots\omega _nx_1}
\end{equation*}
If $A\subseteq\Omega _n$ we define $(A\to )\subseteq\Omega _{n+1}$ by
\begin{equation*}
(A\to )=\cup\brac{(\omega\to )\colon\omega\in A}
\end{equation*}
Thus, $(A\to )$ is the set of one-element continuations of $n$-paths in $A$.

The set of all paths beginning with $\omega\in\Omega _n$ is called an \textit{elementary cylinder set} and is denoted by
$\rmcyl (\omega )$. If $A\subseteq\Omega _n$, then the \textit{cylinder set} $\rmcyl (A)$ is defined by
\begin{equation*}
\rmcyl (A)=\cup\brac{\rmcyl (\omega )\colon\omega\in A}
\end{equation*}
Using the notation
\begin{equation*}
\cscript (\Omega _n)=\brac{\rmcyl (A)\colon A\subseteq\Omega _n}
\end{equation*}
we see that
\begin{equation*}
\cscript (\Omega _1)\subseteq\cscript (\Omega _2)\subseteq\cdots
\end{equation*}
is an increasing  sequence of subalgebras of the \textit{cylinder algebra} $\cscript (\Omega )=\cup\cscript (\Omega _n)$. Letting
$\ascript$ be the $\sigma$-algebra generated by $\cscript (\Omega )$, we have that $(\Omega ,\ascript )$ is a measurable space. For
$A\subseteq\Omega$ we define the sets $A^n\subseteq\Omega _n$ by
\begin{equation*}
A^n=\brac{\omega _1\omega _2\cdots\omega _n\colon\omega _1\omega _2\cdots\omega _n\omega _{n+1}\cdots\in A}
\end{equation*}
That is, $A^n$ is the set of $n$-paths that can be continued to a path in $A$. We think of $A^n$ as the $n$-step approximation to $A$. We have that
\begin{equation*}
\rmcyl (A_1)\supseteq\rmcyl (A_2)\supseteq\cdots\supseteq A
\end{equation*}
so that $A\subseteq\cap\rmcyl (A^n)$. However, $A\ne\cap\rmcyl (A^n)$ in general, even if $A\in\ascript$.

Let $H_n=L_2(\Omega _n)=L_2(\pscript _n)$ be the $n$-\textit{path Hilbert space} $\complex ^{\Omega _n}=\complex ^{\pscript _n}$ with the usual inner product
\begin{equation*}
\elbows{f,g}=\sum\brac{\overline{f(\omega )}g(\omega )\colon\omega\in\Omega _n}
\end{equation*}
For $A\subseteq\Omega _n$, the characteristic function $\chi _A\in H_n$ has norm $\|\chi _A\|=\sqrt{\ab{A}\,}$. In particular
$1_n=\chi _{\Omega _n}$ satisfies
\begin{equation*}
\|1_n\|=\ab{\Omega _n}^{1/2}=2^{(n-1)/2}
\end{equation*}
A positive operator $\rho$ on $H_n$ that satisfies $\elbows{\rho 1_n,1_n}=1$ is called a \textit{probability operator} \cite{gud13}. Corresponding to a probability operator $\rho$ we define the \textit{decoherence functional} \cite{gud13,hen09,sor07}
\begin{equation*}
D_\rho\colon 2^{\Omega _n}\times 2^{\Omega _n}\to\complex
\end{equation*}
by $D_\rho (A,B)=\elbows{\rho\chi _B,\chi _A}$. We interpret $D_\rho (A,B)$ as a measure of the interference between the events $A$ and $B$ when the system is described by $\rho$. We also define the $q$-\textit{measure} $\mu _\rho\colon 2^{\Omega _n}\to\real ^+$ by $\mu _\rho (A)=D_\rho (A,A)$ and interpret $\mu _\rho (A)$ as the quantum propensity of the event $A\subseteq\Omega _n$
\cite{gud13,sor94,sur11}. In general, $\mu _\rho$ is not additive on $2^{\Omega _n}$ so $\mu _\rho$ is not a measure. However,
$\mu _\rho$ is \textit{grade}-2 \textit{additive} \cite{gud13,sor94,sur11} in the sense that if $A,B,C\in 2^{\Omega _n}$ are mutually disjoint, then
\begin{equation*}
\mu _\rho (A\cup B\cup C)=\mu _\rho (A\cup B)+\mu _\rho (A\cup C)+\mu _\rho (B\cup C)
  -\mu _\rho (A)-\mu _\rho (B)-\mu _\rho (C)
\end{equation*}
Let $\rho _n$ be a probability operator on $H_n$, $n=1,2,\ldots\,$. We say that the sequence $\brac{\rho _n}$ is \textit{consistent} if
\begin{equation*}
D_{\rho _{n+1}}(A\to ,B\to )=D_{\rho _n}(A,B)
\end{equation*}
for all $A,B\subseteq\Omega _n$ \cite{gud13}. We call a consistent sequence $\brac{\rho _n}$ a \textit{covariant quantum sequential growth process} (CQSGP). Let $\rho _n$ be a CQSGP and denote the corresponding $q$-measure by $\mu _n$. A set $A\in\ascript$ is \textit{suitable} if $\lim\mu _n(A^n)$ exists (and is finite) in which case we define $\mu (A)=\lim\mu _n(A^n)$. We denote the collection of suitable sets by $\sscript (\Omega )$. Of course, $\emptyset ,\Omega\in\sscript (\Omega )$ with $\mu (\emptyset )=0$,
$\mu (\Omega =1$. If $A\in\cscript (\Omega )$ and $A=\rmcyl (B)$ where $B\subseteq\Omega _m$, then it follows from consistency that $\lim\mu _n(A^n)=\mu _m(B)$. Hence, $A\in\sscript (\Omega )$ and $\mu (A)=\mu _m(B)$. We conclude that 
$\cscript (\Omega )\subseteq\sscript (\Omega )\subseteq\ascript$ and it can be shown that the inclusions are proper, in general. In a sense, $\mu$ is a $q$-measure on $\sscript (\Omega )$ that extends the $q$-measures $\mu _n$.

There are physically relevant sets that are not in $\cscript (\Omega )$. In this case, it is important to know whether such a set $A$ is in
$\sscript (\Omega )$ and if it is, to find $\mu (A)$. For example, if $\omega\in\Omega$ then
\begin{equation*}
\brac{\omega}=\bigcap _{n=1}^\infty\brac{\omega}^n\in\ascript
\end{equation*}
but $\brac{\omega}\notin\cscript (\Omega )$. As another example, the complement $\brac{\omega}'\notin\cscript (\Omega )$. Even if
$\brac{\omega}\in\sscript (\Omega )$, since $\mu _n(A')\ne 1-\mu _n(A)$ for $A\subseteq\Omega _n$ in general, it does not follow immediately that $\brac{\omega}'\in\sscript (\Omega )$. For this reason, we would have to treat $\brac{\omega}'$ as a separate case.

We saw in Section~2 that we can represent each element of $\pscript$ uniquely as $x_{n,\junderbar}$ where $n=\ab{x}$ and
$\junderbar$ can be considered as a binary number. We can also represent each element in $\pscript _{n+1}$ as a $n$-bit binary number $\junderbar =j_nj_{n-1}\cdots j_2j_1$, $j=0$ or $1$. Since $\Omega _n\approx\pscript _n$ we can also represent each
$\omega\in\Omega _{n+1}$ by an $n$-bit binary number $\junderbar$. The standard basis for $H_{n+1}=L_2(\Omega _{n+1})$ is the set of vectors $e_{\junderbar}=\chi _{\omega _{\junderbar}}$, $\omega _{\junderbar}\in\Omega _{n+1}$. We frequently use the notation
$\ket{\junderbar}=e_{\junderbar}$ which is called the \textit{computational basis} in quantum computation theory. In this theory
$\ket{\junderbar}$ is represented by
\begin{equation*}
\ket{\junderbar}=\ket{j_n\cdots j_2j_1}=\ket{j_n}\otimes\cdots\otimes\ket{j_2}\otimes\ket{j_1}
\end{equation*}
where $\ket{j_i}$ is $\ket{0}$ or $\ket{1}$ which form the basis of the two-dimensional Hilbert space $\complex ^2$.

The basis vectors $\ket{0}$ and $\ket{1}$ are called \textit{qubit states} but we shall call them \textit{qubits}, for short. We also call $\ket{\junderbar}$ given above, an $n$-\textit{qubit}. This is the quantum computation analogue of an $n$-bit of classical computer science. If $\rho _{n+1}$ is a probability operator, the corresponding \textit{decoherence matrix} is the $2^n\times 2^n$ complex matrix whose $\junderbar -\kunderbar$ component is given by
\begin{equation*}
M_{\rho _{n+1}}=\sqbrac{\elbows{\rho _n\ket{\kunderbar},\ket{\junderbar}}}
\end{equation*}
This is frequently shortened to
\begin{equation*}
M_{\rho _{n+1}}=\sqbrac{\elbows{\rho _n\kunderbar,\junderbar}}
\end{equation*}
but we shall not use this notation because it can be confusing. For $A,B\subseteq\Omega _{n+1}$ we form the superpositions
\begin{align*}
\ket{A}&=\sum\brac{\ket{\iunderbar}\colon\omega _{\iunderbar}\in A}\\
\ket{B}&=\sum\brac{\ket{\iunderbar}\colon\omega _{\iunderbar}\in B}
\end{align*}
The decoherence functional is now given by
\begin{equation*}
D_{\rho _{n+1}}(A,B)=\elbows{\rho _{n+1}\ket{B},\ket{A}}
\end{equation*}
Superpositions are a strictly quantum phenomenon that has no counterpart in classical computation.

An event $A\subseteq\Omega _n$ is \textit{precluded} if $\mu _n(A)=0$ \cite{sor94}.  Precluded events have been extensively studied in \cite{gtw9,hen09,sor03,sur11,wal13} and they are considered to be events that never occur. We shall give simple examples later which show that if $A$ is precluded and $B\subseteq A$ then $B$ need not be precluded. However, the following properties do hold.

\begin{thm}       
\label{thm31}
{\rm (a)}\enspace If $A\subseteq\Omega _n$ is precluded and $B\subseteq\Omega _n$ is disjoint from $A$ then
$\mu _n(A\cup B)=\mu _n (B)$.
{\rm (b)}\enspace If $A,B\subseteq\Omega _n$ are disjoint precluded events then $A\cup B$ is precluded.
\end{thm}
\begin{proof}
(a)\enspace Since $\mu _n(A)=0$ we have that
\begin{equation*}
\|\rho _n^{1/2}\chi _A\|^2=\elbows{\rho _n^{1/2}\chi _A,\rho _n^{1/2}\chi _A}=\elbows{\rho _n\chi _A,\chi _A}=0
\end{equation*}
Hence, $\rho _n^{1/2}\chi _A=0$ so $\rho _n\chi _A=0$. Since $A\cap B=\emptyset$ we have that
\begin{align*}
\mu _n(A\cup B)&=\elbows{\rho _n\chi _{A\cup B},\chi _{A\cup B}}=\elbows{\rho _n(\chi _A+\chi _B),\chi _A+\chi _B}\\
  &=\elbows{\rho _n\chi _A,\chi _A}+2\rmre\elbows{\rho _n\chi _A,\chi _A}+\elbows{\rho _n\chi _B,\chi _B}\\
  &=\elbows{\rho _n\chi _B,\chi _B}=\mu _n(B)
\end{align*}
Part (b) follows from (a).
\end{proof}

An event $A\in\sscript (\Omega )$ is \textit{precluded} if $\mu (A)=0$ and $A$ is \textit{strongly precluded} if there exists an
$n\in\positive$ such that $\mu _m(A^m)=0$ for all $m\ge n$. For example, if $A=\rmcyl (B)$ where $B\subseteq\Omega _n$ and
$\mu _n(B)=0$ then $A$ is strongly precluded. Of course, strongly precluded events are precluded.

A precluded event is \textit{primitive} if it has no proper, nonempty precluded subsets.

\begin{thm}       
\label{thm32}
If $A\subseteq\Omega _n$ is precluded, then $A$ is primitive or $A$ is a union of mutually disjoint primitive precluded events.
\end{thm}
\begin{proof}
If $A$ is primitive we are finished. Otherwise, there exists a proper, nonempty precluded subset $B\subseteq A$. Since $\ab{B}<\infty$ there exists a nonempty, primitive precluded event $A_1\subseteq B\subseteq A$. Applying Theorem~\ref{thm31}, we conclude that
$\mu _n(A\cap A'_1)=0$. In a similar way, there exists a nonempty, primitive precluded event $A_2\subseteq A\cap A'_1$. Of course, $A_1\cap A_2=\emptyset$. Continuing, this process must eventually stop and we obtain a sequence of mutually disjoint primitive preluded events $A_1,\ldots ,A_n$ with $A=\cup A_i$.
\end{proof}

\section{Covariant Amplitude Processes} 
This section considers a method of constructing a CQSGP called a covariant amplitude process. Not all CQSGPs can be constructed in this way, but this method appears to have physical motivation \cite{gud13}.

A \textit{transition amplitude} is a map $\atilde\colon\pscript\times\pscript\to\complex$ such that $\atilde (x,y)=0$ if $x\not\to y$ and
$\sum _y\atilde (x,y)=1$ for all $x\in\pscript$. This is similar to a Markov chain except $\atilde (x,y)$ may be complex. The
\textit{covariant amplitude process} (CAP) corresponding to $\atilde$ is given by the maps $a_n\colon\Omega _n\to\complex$ where
\begin{equation*}
a_n(\omega _1\omega _2\cdots\omega _n)
   =\atilde (\omega _1,\omega _2)\atilde (\omega _2,\omega _3)\cdots\atilde (\omega _{n-1},\omega _n)
\end{equation*}
We can consider $a_n$ to be a vector in $H_n=L_2(\Omega _n)=L_2(\pscript _n)$. Notice that for $x\in\pscript _n$ we can define $a_n(x)$ to be $a_n(\omega )$ where $\omega\in\Omega _n$ is the unique history of $x$. Observe that
\begin{equation*}
\elbows{1_n,a_n}=\sum _{\omega\in\Omega _n}a_n(\omega )=1
\end{equation*}
and we also have that
\begin{equation*}
\|a_n\|=\paren{\sum _{\omega\in\Omega _n}\ab{a_n(\omega )^2}}^{1/2}
\end{equation*}
Define the rank-1 positive operator $\rho _n=\ket{a_n}\bra{a_n}$ on $H_n$. The norm of $\rho _n$ is
\begin{equation*}
\|\rho _n\|=\|a_n\|^2=\sum _{\omega\in\Omega _n}\ab{a_n(\omega )}^2
\end{equation*}
Since $\elbows{\rho _n1_n,1_n}=\ab{\elbows{1_n,a_n}}^2=1$, we conclude that $\rho _n$ is a probability operator. It is shown in \cite{gud13} that $\brac{\rho _n}$ is consistent so $\brac{\rho _n}$ forms a CQSGP. We call $\brac{\rho _n}$ the CQSGP \textit{generated} by the CAP $\brac{a_n}$.

The decoherence functional corresponding to the CAP $\brac{a_n}$ becomes
\begin{align*}
D_n(A,B)&=\elbows{\rho _n\chi _B,\chi _A}=\elbows{\chi _B,a_n}\elbows{a_n,\chi _A}\\
  &=\sum _{\omega\in A}\overline{a_n(\omega )}\sum _{\omega\in B}a_n(\omega )
\end{align*}
In particular, for $\omega ,\omega '\in\Omega _n$ the decoherence matrix elements
\begin{equation*}
D_n(\omega ,\omega ')=\overline{a_n(\omega )}a_n(\omega ')
\end{equation*}
are the matrix elements of $\rho _n$ in the standard basis. The $q$-measure $\mu _n\colon 2^{\Omega _n}\to\real ^+$ is given by
\begin{equation*}
\mu _n(A)=D_n(A,A)=\ab{\sum _{\omega\in A}a_n(\omega )}^2
\end{equation*}
In particular, $\mu _n(\omega )=\ab{a_n(\omega )}^2$ for every $\omega\in\Omega _n$ and $\mu _n(\Omega _n)=1$. Of course, we also have that $\mu _n(x)=\ab{a_n(x)}^2$ for all $x\in\pscript _n$.

Since each $x\in\pscript _n$ has precisely two offspring, we can describe a transition amplitude $\atilde$ and the corresponding CAP
$\brac{a_n}$ in a simple way. Let
\begin{align*}
\atilde (x_{n,\junderbar},x_{n+1,\junderbar 0})&=c_{n,\junderbar}\\
\intertext{and}
\atilde (x_{n,\junderbar},x_{n+1,\junderbar 1})&=1-c_{n,\junderbar}
\end{align*}
$j=0,1,\ldots ,2^{n-1}-1$. We call the numbers $c_{n,\junderbar}\in\complex$ \textit{coupling constants} for the corresponding CAP
$\brac{a_n}$.
\begin{exam}{3}  
If the CAP $\brac{a_n}$ has coupling constants $c_{n,\junderbar}$, then we have $a_2(x_{2,0})=c_{1,0}$, $a_2(x_{2,1})=1-c_{1,0}$,
$a_3(x_{3,00})=c_{1,0}c_{2,0}$, $a_3(x_{3,01})=c_{1,0}(1-c_{2,0})$, $a_3(x_{3,10})=(1-c_{1,0})c_{2,1}$,
$a_3(x_{3,11})=(1-c_{1,0})(1-c_{2,1})$.
\end{exam}

We shall only need a special case of the next theorem but it still has independent interest.
\begin{thm}       
\label{thm41}
An operator $M$ on $H_n$ is a rank-1 probability operator if and only if $M$ has a matrix representation
$M=\sqbrac{\alphabar _j\alpha _k}$ where $\alpha _j\in\complex$, $j=1,\ldots ,n$, satisfy $\sum\alpha _j=1$.
\end{thm}
\begin{proof}
Suppose $M=\sqbrac{M_{jk}}$ with $M_{jk}=\alphabar _j\alpha _k$ where $\sum\alpha _j=1$. Let $\psi\in\complex ^n$ be the vector
$\psi =(\alpha _1,\ldots ,\alpha _n)$. We have that $M=\ket{\psi}\bra{\psi}$ so $M$ is positive with rank 1. To show that $M$ is a probability operator we have
\begin{equation*}
\sum _{j,k}M_{jk}=\sum _{j,k}\alphabar _j\alpha _k=\ab{\sum\alpha _j}^2=1
\end{equation*}
Conversely, let $M$ be a rank-1 probability operator. Since $M$ is rank-1, it has the form $M=\ket{\psi}\bra{\psi}$ for some
$\psi\in\complex ^n$. We then have the matrix representation
\begin{equation*}
M=\sqbrac{\elbows{e_j,\psi}\elbows{\psi ,e_k}}
\end{equation*}
where $e_j$, $j=1,\ldots ,n$, is the standard basis for $\complex ^n$. Letting $\alpha _j=\elbows{\psi ,e_j}$ we conclude that
$M=\sqbrac{\alphabar _j\alpha _k}$. Since $M$ is a probability operator we have that
\begin{equation*}
\ab{\sum _j\alpha _j}^2=\sum _{j,k}\alphabar _j\alpha _k=1
\end{equation*}
Now there exists a $\theta\in\real$ such that $e^{-i\theta}\sum\alpha _j=1$. Letting $\psi '=e^{i\theta}\psi$ we obtain
\begin{equation*}
M=\ket{\psi '}\bra{\psi '}=\sqbrac{\alphabar '_j\alpha '_k}
\end{equation*}
where $\alpha '_j=e^{-i\theta}\alpha _j$, $j=1,\ldots ,n$. Hence, $\sum\alpha '_j=1$.
\end{proof}

An operator on $\complex ^2$ is called a \textit{qubit operator}. We shall only need the following corollary of Theorem~\ref{thm41}.

\begin{cor}       
\label{cor42}
A qubit operator $M$ is a rank-1 probability operator if and only if $M$ has a matrix representation
\begin{equation}         
\label{eq41}
M=\left[\begin{matrix}\noalign{\smallskip}\ab{c}^2&\cbar (1-c)\\c(1-\cbar )&\ab{1-c}^2\\\end{matrix}\right]
\end{equation}
where $\alpha\in\complex$.
\end{cor}

A CAP $\brac{a_n}$ is \textit{stationary} if the coupling constants $c_{n,\junderbar}$ are independent of $j$. In this case we write
$c_{n,\junderbar}=c_n$ and we have 
$\atilde (x_{n,\junderbar},x_{n+1,\junderbar 0})=c_n,\atilde (x_{n,\junderbar},x_{n+1,\junderbar 1})=1-c_n$. By Corollary~\ref{cor42} the operators
\begin{equation*}
\chat _j=\left[\begin{matrix}\noalign{\smallskip}\ab{c_j}^2&\cbar _j(1-c_j)\\
\noalign{\smallskip}c_j(1-\cbar _j)&\ab{1-c_j}^2\\\end{matrix}\right]
\end{equation*}
are qubit rank-1 probability operators.

\begin{thm}       
\label{thm43}
Let $\brac{c_n}$ be the coupling constants for a stationary CAP. The generated CQSGP $\brac{\rho _n}$ has the form
\begin{equation}         
\label{eq42}
\rho _n=\chat _{n-1}\otimes\chat _{n-2}\otimes\cdots\otimes\chat _2\otimes\chat _1
\end{equation}
\end{thm}
\begin{proof}
Since $\Omega _n\approx\pscript _n$ we can write
\begin{equation*}
\rho _2=D_2=\left[\begin{matrix}\noalign{\smallskip}\ab{c_1}^2&\cbar _1(1-c_1)\\
\noalign{\smallskip}c_1(1-\cbar _1)&\ab{1-c_1}^2\\\end{matrix}\right]=\chat _1
\end{equation*}
At the next step we apply Example~3 to obtain
\begin{align*}
&\rho _3=D_3\\
&=\left[\begin{matrix}\noalign{\smallskip}
\ab{c_1}^2\ab{c_2}^2&\ab{c_1}^2\cbar _2(1-c_2)&\cbar _1(1-c_1)\ab{c_2}^2&\cbar _1(1-c_1)\cbar _2(1-c_2)\\
\noalign{\smallskip}\ab{c_1}^2c_2(1-\cbar _2)&\ab{c_1}^2\ab{1-c_2}^2&\cbar _1(1-c_1)c_2(1-\cbar _2)&\cbar _1(1-c_1)\ab{1-c_2}^2\\
\noalign{\smallskip}c_1(1-\cbar _1)\ab{c_2}^2&c_1(1-\cbar _1)\cbar _2(1-c_2)&\ab{1-c_1}^2\ab{c_2}^2&\ab{1-c_1}^2\cbar _2(1-c_2)\\
\noalign{\smallskip}c_1(1-\cbar _1)c_2(1-\cbar _2)&c_1(1-\cbar _1)\ab{1-c_2}^2&\ab{1-c_1}^2c_2(1-\cbar _2)
  &\ab{1-c_1}^2\ab{1-c_2}^2\\\end{matrix}\right]\\\noalign{\smallskip}
&=\left[\begin{matrix}\noalign{\smallskip}\ab{c_1}^2\chat _2&\cbar _1(1-c_1)\chat _2\\
\noalign{\smallskip}c_1(1-\cbar _1)\chat _2&\ab{1-c_1}^2\chat _2\\\end{matrix}\right]=\chat _2\otimes\chat _1
\end{align*}
Continuing by induction, we have that \eqref{eq42} holds.
\end{proof}

Equation~\eqref{eq42} shows that the $(n-1)$-qubit probability operator $\rho _n$ is the tensor product of $n-1$ qubit probability operators. The next result show that the converse of Theorem~\ref{thm43} holds.

\begin{thm}       
\label{thm44}
If the CQSGP $\brac{\rho _n}$ has the form
\begin{equation*}
\rho _n=\beta _{n-1}\otimes\beta _{n-2}\otimes\cdots\otimes\beta _2\otimes\beta _1
\end{equation*}
where $\beta _j$ is a rank-1 probability operator, then $\brac{\rho _n}$ is generated by a stationary CAP.
\end{thm}
\begin{proof}
Since $\beta _j$, $j=1,\ldots ,n-1$, is a rank-1 qubit probability operator, by Corollary~\ref{cor42} we have that
\begin{equation*}
\beta _j=\left[\begin{matrix}\noalign{\smallskip}\ab{c_j}^2&\cbar _j(1-c_j)\\
\noalign{\smallskip}c_j(1-\cbar _j)&\ab{1-c_j}^2\\\end{matrix}\right]
\end{equation*}
where $c _j\in\complex$. As in the proof of Theorem~\ref{thm43}, $\brac{\rho _n}$ is generated by a stationary CAP whose coupling constants are $\brac{c_n}$
\end{proof}

We say that a CAP is \textit{completely stationary} if the coupling constants $c_{n,\junderbar}$ are independent of $n$ and
$\junderbar$. In this case, we have a single coupling constant $c\in\complex$ and the generated CQSGP $\brac{\rho _n}$ has the form
\begin{equation*}
\rho _n=\chat\otimes\cdots\otimes\chat = \bigotimes _1^{n-1}\chat
\end{equation*}
where $\chat$ has the form \eqref{eq41}.

\section{Examples of $Q$-Measures} 
In this section we compute some simple examples of $q$-measures in the stationary case. Let $\brac{a_n}$ be a stationary CAP with corresponding coupling constants $\brac{c_n}$. As usual, we can identify $\Omega _n$ with $\pscript _n$. If
$\omega =\omega _1\cdots\omega _n\in\Omega _n$ we have that $\mu _n(\omega )=\mu _n(\omega _n)$. For
$\pscript _2=\brac{x_{2,0},x_{2,1}}$ we have $a_2(x_{2,0})=c_1$, $a_2(x_{2,1})=1-c_1$, so $\mu _2(x_{2,0})=\ab{c_1}^2$ and
$\mu _2(x_{2,1})=\ab{1-c_1}^2$. For
\begin{equation*}
\pscript _3=\brac{x_{3,0},x_{3,1},x_{3,2},x_{3,3}}
\end{equation*}
we have $a_3(x_{3,0})=c_1c_2$, $a_3(x_{3,1})=c_1(1-c_2)$, $a_3(x_{3,2})=(1-c_1)c_2$, $a_3(x_{3,3})=(1-c_1)(1-c_2)$. Hence, 
$\mu _3(x_{3,0})=\ab{c_1}^2\ab{c_2}^2$, $\mu _3(x_{3,1})=\ab{c_1}^2\ab{1-c_2}^2$, $\mu _3(x_{3,2})=\ab{1-c_1}^2\ab{c_2}^2$ and
$\mu _3(x_{3,3})=\ab{1-c_1}^2\ab{1-c_2}^2$. We now compute the $q$-measure of some two element sets. We have that
\begin{equation*}
\mu _3\paren{\brac{x_{3,0},x_{3,1}}}=\ab{a_3(x_{3,0})+a_3(x_{3,1})}^2=\ab{c_1}^2
\end{equation*}
Since $\mu _3\paren{\brac{x_{3,0},x_{3,1}}}\ne\mu _3(x_{3,0})+\mu _3(x_{3,1})$ in general, we conclude that $x_{3,0}$ and $x_{3,1}$ interfere with other, except in special cases. If
\begin{equation*}
\mu _3\paren{\brac{x_{3,0},x_{3,1}}}<\mu _3(x_{3,0})+\mu _3(x_{3,1})
\end{equation*}
we say that $x_{3,0}$ and $x_{3,1}$ interfere \textit{destructively} and if
\begin{equation*}
\mu _3\paren{\brac{x_{3,0},x_{3,1}}}>\mu _3(x_{3,0})+\mu _3(x_{3,1})
\end{equation*}
we say that $x_{3,0}$ and $x_{3,1}$ interfere \textit{constructively}. The three possible cases, $=,<,>$ can occur depending on the value of $c_2$. In a similar way, we have that $\mu\paren{\brac{x_{3,0},x_{3,2}}}=\ab{c_2}^2$
\begin{align*}
\mu _3\paren{\brac{x_{3,0},x_{3,3}}}&=\ab{1-c_1-c_2+2c_1c_2}^2\\
\mu _3\paren{\brac{x_{3,1},x_{3,2}}}&=\ab{c_1+c_2-2c_1c_2}^2\\
\mu _3\paren{\brac{x_{3,1},x_{3,3}}}&=\ab{1-c_2}^2
\end{align*}
It follows that any pair of elements of $\pscript _3$ interfere, in general. Finally, we compute the $q$-measures of some three element sets:
\begin{align*}
\mu\paren{\brac{x_{3,0},x_{3,1},x_{3,2}}}&=\ab{c_1+c_2-2c_1c_2}^2\\
\mu _3\paren{\brac{x_{3,0},x_{3,1},x_{3,3}}}&=\ab{1-c_2+2c_1c_2}^2
\end{align*}

We now consider
\begin{equation*}
\pscript _4=\brac{x_{4,0},x_{4,1},x_{4,2},x_{4,3},x_{4,4},x_{4,5},x_{4,6},x_{4,7}}
\end{equation*}
In this case we have $a_4(x_{4,0})=c_1c_2c_3$, $a_4(x_{4,1})=c_1c_2(1-c_3)$, $a_4(x_{4,2})=c_1(1-c_2)c_3$,
$a_4(x_{4,3})=c_1(1-c_2)(1-c_3)$, $a_4(x_{4,4})=(1-c_1)c_2c_3$, $a_4(x_{4,5})=(1-c_1)c_2(1-c_3)$,
$a_4(x_{4,6})=(1-c_1)(1-c_2)c_3$, $a_4(x_{4,7})=(1-c_1)(1-c_2)(1-c_3)$. We then have that $\mu _4(x_{4,j})=\ab{a_4(x_{4,j})}^2$, $j=0,1,\ldots ,7$. In general, the pattern is clear that
\begin{equation*}
\mu _n(x_{n,j})=\ab{c'_1}^2\ab{c'_2}^2\cdots\ab{c'_{n-1}}^2
\end{equation*}
where $c'_i=c_i$ if the history of $x_{n,j}$ turns ``left'' at the $i$th step and $c'_i=1-c_i$ if it turns ``right'' at the $i$th step. Some
$q$-measures of two element sets are
\begin{align*}
\mu _4\paren{\brac{x_{4,0},x_{4,1}}}&=\ab{c_1}^2\ab{c_2}^2\\
\mu _4\paren{\brac{x_{4,1},x_{4,2}}}&=\ab{c_1}^2\ab{c_2+c_3}^2
\end{align*}
In general, any pair of $c$-causets in $\pscript _4$ interfere.

We now consider the \textit{extremal left path} $\omega _\ell =x_{1,0}x_{2,0}x_{3,0}\cdots\,$.
Is $\brac{\omega _\ell}\in\sscript (\Omega )$? We have that
\begin{equation*}
\mu _n(x_{n,0})=\ab{c'_1}^2\ab{c'_2}^2\cdots\ab{c'_{n-1}}^2
\end{equation*}
Now $\brac{\omega _\ell}\in\sscript (\Omega )$ if and only if $\lim\mu _n(x_{n,0})$ exists and this depends on the values of $c_n$. In fact, we can set values of $c_n$ so that $\lim\mu _n(x_{n,0})=r$ for an $r\in\real ^+$. For example, if we let $c_n=c_{n+1}=\cdots =1$, then we obtain
\begin{equation*}
\mu (\omega _\ell)=\lim\mu _m(x_{m,0})=\ab{c'_1}^2\ab{c'_2}^2\cdots\ab{c'_{n-1}}^2
\end{equation*}
Moreover, in this case $\brac{\omega}\in\sscript (\Omega )$ for every $\omega\in\Omega$ with similar values for $\mu (\omega )$.

As another example, let $A\subseteq\Omega$ be the set of paths $\omega =\omega _1\omega _2\cdots$ such that
$\omega _3,\omega _4,\cdots$ are the ``middle half'' of $\pscript _3,\pscript _4,\ldots\,$. That is, 
$A^1=\pscript _1$, $A^2=\pscript$, $A^3=\brac{x_{3,1},x_{3,2}}$,
\begin{align*}
A^4&=\brac{x_{4,2},x_{4,3}x_{4,4},x_{4,5}}\\
A^5&=\brac{x_{5,4},x_{5,5}x_{5,6},x_{5,7}x_{5,8}x_{5,9}x_{5,10}x_{5,11}}\\
\vdots\ \ &
\end{align*}
Now $\mu _1(A^1)=\mu _2(A^2)=1$, $\mu _3(A^3)=\ab{c_1+c_2-2c_1c_2}^2$
\begin{align*}
\mu _4(A^4)&=\ab{c_1(1-c_2)c_3+c_1(1-c_2)(1-c_3)+(1-c_1)c_2c_3+(1-c_1)c_2(1-c_3)}^2\\
&=\ab{c_1(1-c_2)+(1-c_1)c_2}^2=\ab{c_1+c_2-2c_1c_2}^2
\end{align*}
It is not a coincidence that $\mu _4(A^4)=\mu _3(A^3)$. In fact, $A^4=(A^3\to )$ and $A=\rmcyl (A^3)$. It follows that $A\in\cscript (\Omega )$ so $A\in\sscript (\Omega )$ with $\mu (A)=\mu _3(A^3)$. In a similar way $A'\in\sscript (\Omega )$ with
$\mu (A')=\ab{1-c_1-c_2+2_1c_2}^2$. We can interpret $A'$ as the ``one fourth end paths'' with $A'^{n}=(A^{n})'$, $n=3,4,\ldots\,$.

The situation for noncylinder sets is more complicated so to simplify matters we consider a completely stationary CAP. In this case we have only one coupling constant $c$. For $x\in\pscript _n$ we have that $a_n(x)=c^j(1-c)^k$ where $j+k=n-1$, $j$ is the number of ``left turns'' and $k$ is the number of ``right turns.'' We then have explicitly that 
\begin{align*}
\mu _n(\pscript _n)=\ab{\sum _{x\in\pscript _n}a_n(x)}^2
  &=\ab{\sum _{j=0}^{n-1}\begin{pmatrix}n-1\\j\end{pmatrix}c^j(1-c)^{(n-1-j)}}^2\\
  &=\ab{(c+1-c)^{n-1}}^2=1
\end{align*}
The $q$-measure of $x\in\pscript _n$ becomes
\begin{equation*}
\mu _n(x)=\ab{c^j(1-c)^k}^2=\ab{c}^{2j}\ab{1-c}^{2k}
\end{equation*}
It is interesting that in this case we have
\begin{align*}
\sum _{x\in\pscript _n}\mu _n(x)&=\sum _{j+0}^{n-`}\begin{pmatrix}n-1\\j\end{pmatrix}\paren{\ab{c}^2}^j\paren{\ab{1-c}^2)^{(n-1-j)}}\\
  &=\paren{\ab{c}^2+\ab{1-c}^2}^{n-1}
\end{align*}
If $\omega =\omega _1\omega _2\cdots\in\Omega$, then
\begin{equation*}
\mu _n\paren{\brac{\omega}^n}=\mu _n(\omega _n)=\ab{c}^{2j}\ab{1-c}^{2k}
\end{equation*}
Whether $\lim\mu _n(\omega _n)$ exists or not depends on $c$. If $\ab{c},\ab{1-c}<1$ then $\brac{\omega}\in\sscript (\Omega )$ for every $\omega\in\Omega$ and $\mu (\omega )=0$. If $\ab{c},\ab{1-c}>1$, then $\brac{\omega}\notin\sscript (\Omega )$ for every
$\omega\in\Omega$. If $\ab{c}<1$, $\ab{1-c}>1$ or vice versa, then $\brac{\omega}\in\sscript (\Omega )$ for some $\omega\in\Omega$ and $\brac{\omega}\notin\sscript (\Omega )$ for others. Except for the trivial cases $\ab{c}=1$ or $\ab{1-c}=1$ we have that
$\mu (\brac{\omega})=0$ whenever $\ab{\omega}\in\sscript (\Omega )$. An interesting example of a set $B\notin\cscript (\Omega )$ is
\begin{equation*}
B=\brac{\omega _1\omega _2\cdots\in\Omega\colon\omega _j\hbox{ is connected }j\in\positive}
\end{equation*}
Thus, $B=\brac{\omega _\ell}'$ where $\omega _\ell$ is the extremal left path. Then $B\notin\cscript (\Omega)$ and
$\mu _n(B^n)=\ab{1-c^{n-1}}^2$. If $\ab{c}<1$, then $\lim\mu _n(B^n)=1$ so $B\in\sscript (\Omega )$ with $\mu (B)=1$.

As a special case, let $\brac{a_n}$ be a completely stationary CAP with coupling constant $c=\frac{1}{2}+\frac{i}{2}$. This is probably the simplest nontrivial coupling constant. Notice that $1-c=\frac{1}{2}-\frac{i}{2}=\cbar$ and $\ab{c}=\ab{\cbar}=1/\sqrt{2}$. Moreover
\begin{equation*}
c=\frac{1}{\sqrt{2}}e^{i\pi /4},\quad\cbar =\frac{1}{\sqrt{2}}e^{-i\pi /4}
\end{equation*}
For $x\in\pscript _n$ we have that $\mu _n(x)=1/2^{n-1}$. It follows that $\brac{\omega}\in\sscript (\Omega )$ for every $\omega\in\Omega$ and $\mu (\omega )=0$. In a similar way, if $A\subseteq\Omega$ is finite, then $A\in\sscript (\Omega )$ and $\mu (A)=0$. Moreover, $A'\in\sscript (\Omega )$ and $\mu (A')=1$. In $\pscript _3$ we have that
\begin{equation*}
\mu _3\paren{\brac{x_{3,0},x_{3,1}}}=\ab{c}^2=\tfrac{1}{2}=\mu _3(x_{3,0})+\mu _3(x_{3,1})
\end{equation*}
so in this case $x_{3,0}$ and $x_{3,1}$ do not interfere. In a similar way, $\mu _3\paren{\brac{x_{3,0},x_{3,2}}}=1/2$ so $x_{3,0}$ and $x_{3,2}$ do not interfere. On the other hand,
\begin{equation*}
\mu _3\paren{\brac{x_{3,0},x_{3,3}}}=\ab{1-2c-2c^2}^2=0
\end{equation*}
so $x_{3,0}$ and $x_{3,3}$ interfere destructively. Also,
\begin{equation*}
\mu _3\paren{\brac{x_{3,1},x_{3,2}}}=\ab{2c-2c^2}^2=4\ab{c(1-c)}^2=1
\end{equation*}
so $x_{3,1}$ and $x_{3,2}$ interfere constructively. Even in this simple case we can get strange results:
\begin{equation*}
\mu _3\paren{\brac{x_{3,0},x_{3,1},x_{3,2}}}=\ab{2c-2c^2}^2=\frac{5}{4}
\end{equation*}
We can check grade-2 additivity:
\begin{align*}
\frac{5}{4}&=\mu _3\paren{\brac{x_{3,0},x_{3,1},x_{3,2}}}\\
  &=\mu _3\paren{\brac{x_{3,0},x_{3,1}}}+\mu _3\paren{{x_{3,0},x_{3,2}}}+\mu _3\paren{\brac{x_{3,1},x_{3,2}}}\\
  &\quad -\mu _3(x_{3,0})-\mu _3(x_{3,1})-\mu _3(x_{3,2})=\frac{1}{2}+\frac{1}{2}+1-\frac{3}{4}
\end{align*}
An interesting property of this special case is that the probability operators $\rho _n=\rho _2\otimes\cdots\otimes\rho _2$ are closely related to the Pauli spin operator
\begin{equation*}
\sigma _y=\left[\begin{matrix}\noalign{\smallskip}0&-i\\
\noalign{\smallskip}i&0\\\end{matrix}\right]
\end{equation*}
In particular, for $c=\frac{1}{2}+\frac{i}{2}$ we have
\begin{equation*}
\rho _2=\left[\begin{matrix}\noalign{\smallskip}\ab{c}^2&\cbar (1-c)\\
\noalign{\smallskip}c(1-\cbar )&\ab{1-c}^2\\\end{matrix}\right]
=\frac{1}{2}\left[\begin{matrix}\noalign{\smallskip}1&-i\\\noalign{\smallskip}i&1\\\end{matrix}\right]
=\frac{1}{2}\,(I+\sigma _y)
\end{equation*}
In this way, $\rho _n$ corresponds to a state for $(n-1)$ spin-$\frac{1}{2}$ particles.

We now consider precluded events for the CAP we are discussing. We say that $x_{n,j},x_{n,k}\in\pscript _n$ are an \textit{antipodal pair} if $a_n(x_{n,j})=-a_n(x_{n,k})$. Since $a_n(x_{n,m})=c^j\cbar ^{\,k}$, $j+k=n-1$, we have that
\begin{equation*}
a_n(x_{n,m})=2^{(n-1)/2}e^{ir\pi /4}
\end{equation*}
for sone $r\in\brac{0,1,\ldots ,7}$. It follows that $x_{n,j}$ and $x_{n,k}$ are an antipodal pair if and only if
\begin{equation*}
a_n(x_{n,j})=2^{(n-1)/2}e^{ir\pi /4}=-a_n(x_{n,k})
\end{equation*}
for some $r\in\brac{0,1,\ldots ,7}$. We leave the proof of the following result to the reader. As usual we apply the identity
$\Omega _n\approx\pscript _n$.
 
\begin{thm}       
\label{thm51}
A set $A\subseteq\Omega _n$ is a nonempty, primitive precluded event if and only if $A=\brac{x_{n,j},x_{n,k}}$ where $x_{n,j}$ and $x_{n,k}$ are an antipodal pair.
\end{thm}

Applying Theorems~\ref{thm51} and \ref{thm32} we obtain:
\begin{cor}       
\label{cor52}
A set $A\subseteq\Omega _n$ is precluded if and only if $A$ is a disjoint union of antipodal pairs.
\end{cor}

\begin{exam}{4}  
We illustrate Corollary~\ref{cor52} by displaying the antipodal pairs in $\pscript _3$, $\pscript _4$ and $\pscript _5$. In $\pscript _3$ there is only one antipodal pair $(x_{3,0},x_{3,3})$. In $\pscript _4$ the antipodal pairs are
\begin{align*}
&(x_{4,0},x_{4,3}),(x_{4,0},x_{4,5}),(x_{4,0},x_{4,6})\\
&(x_{4,1},x_{4,7}),(x_{4,2},x_{4,7}),(x_{4,4},x_{4,7})
\end{align*}
In $\pscript _5$ there are 28 antipodal pairs. To save writing we use the notation $j=x_{5,j}$. The antipodal pairs in $\pscript _5$ are
\begin{align*}
&(0,3),(0,5),(0,6),(0,9),(0,10),(0,12)\\
&(15,3),(15,5),(15,6),(15,9),(15,10),(15,12)\\
&(1,7),(1,11),(1,13),(1,14),(2,7),(2,11),(2,13),(2,14)\\
&(4,7),(4,11),(4,13),(4,14),(8.7),(8,11),(8,13),(8,14)
\end{align*}
\end{exam}

According to the coevent formulation \cite{gtw9,sor07,sur11,wal13}, precluded events do not occur so we can remove them from consideration. What is left can occur in some anhomomorphic realization of possible universes \cite{gtw9,sur11,wal13}. We can remove a precluded event from $\Omega _n$ (or $\pscript _n$) which is as large as possible but there is no unique way of doing this, in general. To illustrate this method let us remove the ``left'' and ``right'' precluded extremes. In $\pscript _3$ we remove the precluded event
$\brac{x_{3,0},x_{3,3}}$ and we obtain
\begin{equation*}
A_3=\brac{x_{3,1},x_{3,2}}
\end{equation*}
with $\mu _3(A_3)=1$. In $\pscript _4$ we remove the precluded event
\begin{equation*}
\brac{x_{4,0},x_{4,1},x_{4,6},x_{4,7}}
\end{equation*}
and we obtain
\begin{equation*}
A_4=\brac{x_{4,2},x_{4,3},x_{4,4},x_{4,5}}
\end{equation*}
with $\mu _4(A_4)=1$. In $\pscript _5$ we remove the precluded event
\begin{equation*}
\brac{x_{5,0},x_{5,1},x_{5,2},x_{5,3},x_{5,4},x_{5,11},x_{5,12},x_{5,13},x_{5,14},x_{5,15}}
\end{equation*}
and we obtain
\begin{equation*}
A_5=\brac{x_{5,5},x_{5,6},x_{5,7},x_{5,8},x_{5,9},x_{5,10}}
\end{equation*}
with $\mu _5(A_5)=1$. Continuing this process, we conjecture that we obtain a sequence of events $A_n\subseteq\Omega _n$ where
$\ab{A_n}=2(n-2)$ and $\mu _n(A_n)=1$. Although $\ab{\Omega _n}$ increases exponentially, if this conjecture holds then $\ab{A_n}$ only increases linearly. This gives a huge reduction for the number of possible universes. If $A\subseteq\Omega$ satisfies $A^n=A_n$ then $A\in\sscript (\Omega )$ with $\mu (A)=1$ and $A'\in\sscript (\Omega )$ with $\mu (A')=0$. We would then conclude that $A'$ is precluded and a realizable universe would have to be in $A$.

\end{document}